\documentclass[sn-mathphys-num]{sn-jnl}

\usepackage{tikz}
\usetikzlibrary{positioning, arrows.meta} 
\usepackage{graphicx}%
\usepackage{multirow}%
\usepackage{amsmath,amssymb,amsfonts}%
\usepackage{amsthm}%
\usepackage{mathrsfs}%
\usepackage[title]{appendix}%
\usepackage{xcolor}%
\usepackage{textcomp}%
\usepackage{manyfoot}%
\usepackage{booktabs}%
\usepackage{algorithm}%
\usepackage{algorithmicx}%
\usepackage{algpseudocode}%
\usepackage{listings}%
\RequirePackage{todonotes}
\usepackage{placeins}
\usepackage{subcaption}

\newtheorem{thm}{Theorem}
\newtheorem{lem}[thm]{Lemma}

\newtheorem{defi}[thm]{Definition}

\newcommand{\setX}{\mathcal{X}} 
\newcommand{\setY}{\mathcal{Y}}

\theoremstyle{thmstyleone}%

\begin{document}

\title[Article Title]{Sparse Bayesian Learning for Label Efficiency in Cardiac Real-Time MRI}

\author[1,2]{Felix Terhag} 
\author[1]{Philipp Knechtges}
\author[1]{Achim Basermann}
\author[4]{Anja Bach}
\author[4]{Darius Gerlach}
\author[4]{Jens Tank}
\author[2,3]{Raúl Tempone}

\affil[1]{\orgdiv{Institute of Software Technology, High-Performance Computing}, \orgname{German Aerospace Center (DLR)}, \city{Cologne}, \country{Germany}}

\affil[2]{\orgdiv{Chair of Mathematics for Uncertainty Quantification}, \orgname{Department of Mathematics, RWTH Aachen University}, \city{Aachen}, \country{Germany}}
\affil[3]{\orgdiv{Computer, Electrical and Mathematical Sciences and
Engineering Division (CEMSE)}, \orgname{King Abdullah University of Science and Technology (KAUST)}, \city{Thuwal}, \country{Saudi Arabia}}
\affil[4]{\orgdiv{Institute of Aerospace Medicine}, \orgname{German Aerospace Center (DLR)}, \city{Cologne}, \country{Germany}}

\abstract{Cardiac real-time magnetic resonance imaging (MRI)  is an emerging technology that images the heart at up to 50 frames per second, offering insight into the respiratory effects on the heartbeat. However, this method  significantly increases the number of images that must be segmented to derive critical health indicators. Although neural networks perform well on inner slices, predictions on outer slices are often unreliable.  

This work proposes sparse Bayesian learning (SBL) to predict the ventricular volume on outer slices with minimal manual labeling to address this challenge. The ventricular volume over time is assumed to be dominated by sparse frequencies corresponding to the heart and respiratory rates. Moreover, SBL identifies these sparse frequencies on well-segmented inner slices by optimizing hyperparameters via type -II likelihood, automatically pruning irrelevant components. The identified sparse frequencies guide the selection of outer slice images for labeling, minimizing posterior variance. 

This work provides performance guarantees for the greedy algorithm. Testing on patient data demonstrates that only a few labeled images are necessary for accurate volume prediction. The labeling procedure effectively avoids selecting inefficient images. Furthermore, the Bayesian approach provides uncertainty estimates, highlighting unreliable predictions (e.g.,  when choosing suboptimal labels).}


\maketitle
\section{Introduction}
Recent advances in automated segmentation have improved the analysis of cardiac magnetic resonance imaging (MRI), driven by publicly available datasets \cite{Acdc_2018, MnM_2021} and the U-Net architecture \cite{isensee_nnu-net_2021, Long_FCN_2017}. However, challenges remain because these approaches struggle to generalize to rare pathologies \cite{Acdc_2018, MnM_2021} or different heart regions \cite{terhag2025uq}. Real-time MRI amplifies these problems because it is a new, emerging technology with few publicly available datasets. Real-time MRI enables video sequence acquisition at up to 50 frames per second (fps) for multiple slices throughout the heart. Observing respiratory effects on the heartbeat by removing the need for breath-holding, which is typically required in traditional cine MRI scans, is beneficial for practitioners \cite{rt_mri10,MRI_emerging17,zhang_real-time_2014}. However, real-time MRI complicates manually labeling because it requires annotating an entire sequence rather than a single image. For instance, the study dataset, captured at 30 fps, requires a minimum of 10~s  to record sufficient breathing cycles. With at least 15 slices per heart, this results in over 4,500 images requiring annotation.

This work aims to predict ventricular volumes in real-time cardiac MRI with minimal manual labeling. The MRI scans in this study originated from patients with  univentricular hearts, a rare congenital heart disease where patients are born with only a single functioning ventricle, requiring several surgeries altering the geometry of the heart. A more detailed description of the study that acquired the data can be found in \cite{hypofon_study}. This disease complicates training even further. While U-Nets can segment inner slices with high accuracy, they tend to struggle with the outer slices a phenomenon also observed in classical cine MRI approaches \cite{terhag2025uq}.  Experts validated the segmentations in our dataset, identifying inner slices with negligible errors alongside outer slices with unreliable segmentation, complicating the derivation of key clinical parameters. To prevent a domain expert from labeling the whole video sequences of insufficiently segmented slices, we propose a method which
\begin{enumerate}
    \item allows us to extract information from the well-segmented slices to predict the volume of the bad slices with little labeling effort,
    \item helps us choose the frames to label effectively,
    \item and accounts for the uncertainty within the model.
\end{enumerate}

Sparse Bayesian learning (SBL) \cite{Tipp01} is applied to achieve this, exploiting the periodic nature of the volume signal dominated by the heart and respiratory rates. Sparse frequency components are identified using the type II likelihood to determine the hyperparameters of the prior. Moreover, SBL is widely used in diverse fields, such as remote sensing, where it is effective for finding sparse representations from measurement data~\cite{Gemba_SBL_application, SBL_DOA}. In remote sensing applications, the measurement data usually comprise  complex values, while we adapted the SBL algorithm to real-valued data. The sparse prior facilitates transferring information from well-segmented inner slices to less accurate outer slices, requiring only a few  measurements to fit the phase and amplitude of the sparse frequencies. Sparse frequencies can be used to select frames for labeling while minimizing the posterior uncertainty. The trace of the posterior covariance quantifies the posterior uncertainty. T he objective function is weakly submodular; hence, a lower bound for the suboptimality of the greedy algorithm can be set. Figure~\ref{fig:scheme} provides a schematic representation of the proposed approach.

To our knowledge, this approach has not previously been explored. In remote sensing, adaptive methods typically optimize models with respect to the frequencies to identify off-grid direction-of-arrival sources. An optimization concerning  time points is irrelevant in this framework \cite{wang2_adaptive, Wang_adaptive_doa,adaptive4,adaptive3}. Testing on real patient data demonstrates  that a few labeled examples are sufficient for accurate volume predictions. Furthermore, uncertainty estimates highlight unreliable predictions for suboptimal labels.

The rest of the paper is organized as follows. Section~\ref{sec:methods} explores the adapted SBL algorithm and optimization problem. Additionally, this section proves that the greedy algorithm lies within certain bounds of the optimal solution. Next, Section~\ref{sec:application} introduces the real-time MRI application. Then, Section~\ref{sec:results} presents the results on actual patient data, and Section~\ref{sec:conclusion} concludes this paper.

\begin{figure}
\begin{center}
\begin{tikzpicture}[
  node distance=1.5cm and 0.8cm, %
  box/.style={draw, rounded corners, align=center, minimum width=2.5cm, minimum height=1.8cm}, %
  smallbox/.style={draw, rounded corners, align=center, minimum width=2cm, minimum height=1.2cm}, %
  arrow/.style={thick, ->, >=stealth},
  interrupt/.style={thick, sloped, midway, fill=white, inner sep=1pt}, %
  label below/.style={midway, below, sloped},
  label above/.style={midway, above, sloped}
]
\tiny
\newcommand{\pictosize}{0.7cm}
\node[box] (test1) {
    \includegraphics[width=4cm]{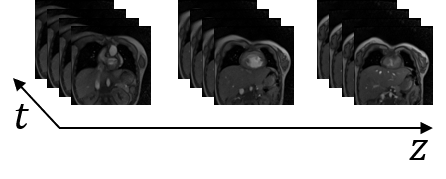} \\
    Cardiac real-time MRI
};

\node[box, below=0.7cm of test1] (test2) {
    \includegraphics[width=4cm]{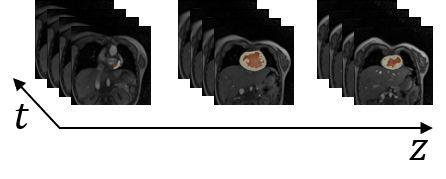} \\
    Segmentations
};

\draw[arrow] (test1) -- node[interrupt, rotate=90] {U-Net} (test2);

\node[below=0.4cm of test2] (pik1) {
    \includegraphics[width=\pictosize]{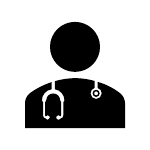}
};
\node[below=0.0cm of pik1] (evaluate) {Evaluate};
\draw[arrow] (test2) -- (pik1);

\newcommand{\lrtikzsize}{2.5cm}
\node[smallbox, left=\lrtikzsize of pik1] (good) {
    \includegraphics[width=2cm]{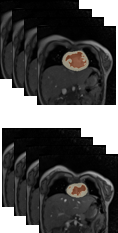} \\
    Good slices
};

\node[smallbox, right=\lrtikzsize of pik1] (bad) {
    \includegraphics[width=2cm]{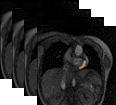} \\
    Bad slices
};

\draw[arrow] (pik1) -- (good);
\draw[arrow] (pik1) -- (bad);

\node[smallbox, align=left,below=0.7cm of good] (example) {
\begin{tabular}{lc}
\multicolumn{2}{l}{\textbf{Calculate sparse prior}}\\
\hline
\textbf{while} $\epsilon > \epsilon_{\min}$:  & \\
\hspace*{2ex}  calculate $\mu_x, \Sigma_x$& \eqref{eq:post_mu_sig} \\
\hspace*{2ex} update $\alpha_m$ & \eqref{eq:alpha_update}\\
\hspace*{2ex} update $\sigma^2$ & \eqref{eq:update_sigma}\\
\end{tabular}
};
\draw[arrow] (good) -- (example);

\node[smallbox,below= 0.7cm of example] (prior) {
    \includegraphics[width=4cm]{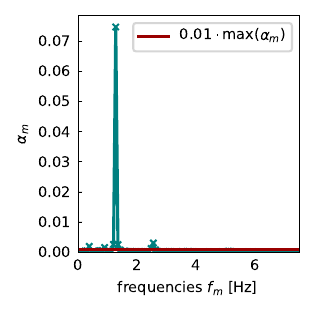}
};
\draw[arrow] (example) --(prior);
\node[below=0.4cm of bad] (next_index) {Draw next Image};
\draw[arrow] (bad) -- (next_index);

\node[below=0.3cm of next_index] (label_image) {
    \includegraphics[width=\pictosize]{figures/pik1.pdf}
};
\draw[arrow] (next_index) -- (label_image);
\node[smallbox,,below=0.7cm of label_image] (test5) {
    \includegraphics[width=1.4cm]{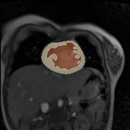}
};

\node[below=0.0cm of label_image] (label_image_text) {Label image};
\draw[arrow] (label_image_text) -- (test5);
\node[smallbox,below=0.7cm of test5] (test6) {
    \includegraphics[width=4cm]{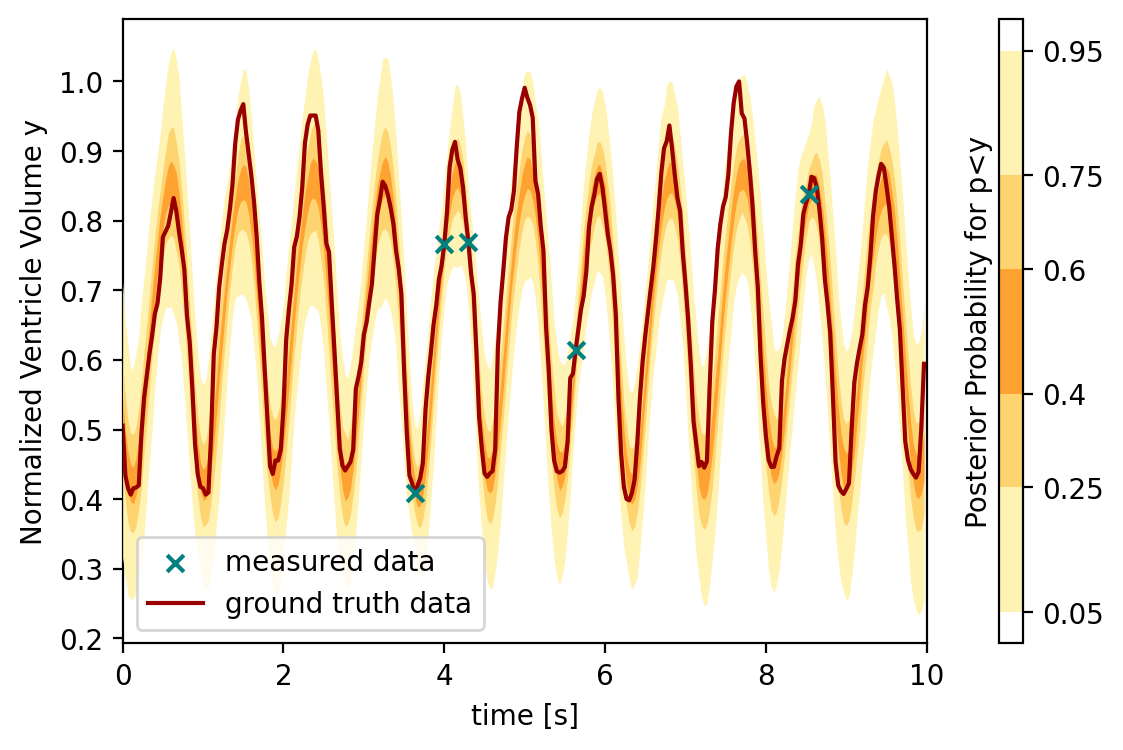}
};
\draw[arrow] (test5) -- node[interrupt, rotate=90] {Update posterior} (test6);

\node[below=0.4cm of test6] (pik3) {
    \includegraphics[width=\pictosize]{figures/pik1.pdf}
};
\node[below=0.cm of pik3] {Evaluate};
\draw[arrow] (test6) -- (pik3);
\newcommand{\tikzsizecross}{1.3cm}
\node[right=\tikzsizecross of pik3] (cross) {
    \includegraphics[width=0.5cm]{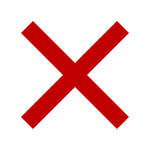}
};

\draw[arrow] (cross.east) -- ++(\tikzsizecross,0) |- (next_index.east);
\draw[thick] (pik3) -- (cross);
\node[smallbox, draw=none, below=3cm of evaluate] (labelnode) {Calculate\\ labeling order\\ following \eqref{eq:greedy_optimum}};

\draw[thick] (prior) -| (labelnode);

\draw[arrow] (labelnode) |- (next_index);
\node[below=5.2cm of labelnode] (check) {
    \includegraphics[width=0.5cm]{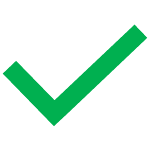}
};
\draw[arrow] (pik3) -| (check);
\end{tikzpicture}
\end{center}
\caption{Scheme of the procedure. A neural network segments the raw images. An expert separates the segmentations into "good" and "bad" slices. The good slices are used to identify the dominating frequencies by empirically deriving a sparse prior in the frequency domain. The prior is used to determine the optimal order to label the images. An expert manually segments the "bad" slices frame by frame until the posterior provides sufficient certainty. The MRIs are short-axis views of univentricular hearts acquired by \cite{hypofon_study}.}
\label{fig:scheme}
\end{figure}
\section{Methods}\label{sec:methods}
\subsection{Modeling}
Cardiac real-time MRI retrieves video sequences of the beating heart due to a faster capture rate than cine MRI. Thus,  for every slice $\ell\in [1,L]$, there are images for all distinct equidistant time  steps $t_1,\ldots,t_N$. Only the ventricle volume and derived parameters, such as stroke volume, are of interest. The volume over time is assumed to be well described by a superposition of frequencies, comprising a subset of $\mathcal{F}=\{f_0=0, f_1, f_2, \ldots , f_{M}\}$, where the main frequencies are the unknown heart and respiratory rates. This work proposes the following linear model:
\begin{equation}\label{eq:lin_model}
    Y=AX+N,
\end{equation} 
with
\begin{itemize}
    \item $Y\in \mathbb{R}^{N\times L}$, the real-valued measurement points $y_{k,\ell}$;
    
    \item $X\in \mathbb{R}^{2M+1\times L}$,  the real and imaginary parts of the complex amplitudes;
    \item $N\in \mathbb{R}^{N\times L}$, additive noise with $n_{k,\ell}\sim \mathcal{N}(0,\sigma^2)$;
    \item $A\in \mathbb{R}^{N\times2M+1}$. the transfer matrix $A=\left[\mathbf{1}_N\ Re(B)\ Im(B) \right]$,
    \item with $\mathbf{1}_N$ an $N$-dimensional vector of 1s and $B\in \mathbb{C}^{N\times M}$ defined as follows,
    \begin{equation}
        B = \begin{bmatrix} 
\exp(-i2\pi f_1t_1) & \dots  & \exp(-i2\pi f_{M}t_{1})\\
    \vdots & \ddots & \vdots \\
    \exp(-i2\pi f_{1}t_{N}) & \dots  & \exp(-i2\pi f_{M}t_{N})
    \end{bmatrix}.
    \end{equation}
\end{itemize}

\subsection{Finding sparse frequencies}\label{ssec:find_freqs}
The above model can describe measurements via superpositioned frequencies of $\mathcal{F}$. Only a few frequencies are assumed to influence the volume substantially. Furthermore, the influential frequencies (e.g.,  heart and respiratory rates) are assumed to be shared over all slices. Identifying these sparse frequencies is beneficial because the model predicts slices where only a few measurement points exist in time . Fewer amplitudes and phases must be found if fewer frequencies exist.

This work applies an approach similar to \cite{SBL_DOA} to determine the sparse frequencies. The procedure is presented in Algorithm~\ref{alg:calc_prior}. In \cite{SBL_DOA}, the authors used  SBL to localize the sparse sources from noisy signals in the signal processing application of estimating the directions of arrival of plane waves from sensor array data. They  obtain  observations from $L$ different snapshots, whereas our method applies the volume over time for $L$ different heart slices. One additional difference of the proposed approach is that the measurements in \cite{SBL_DOA} are complex compared to real-valued measurements.

A Bayesian model is designed to use the SBL  approach. With the assumed additive Gaussian noise of variance $\sigma^2$, the likelihood is
\begin{equation}\label{eq:likelihood}
    p(Y\mid X, \sigma^2) = \frac{\exp\left(-\frac{1}{2\sigma^2}\, \left\Vert AX-Y \right\Vert_F^2\right) }{(2\pi\sigma^2)^{NL/2} }.
\end{equation}
Moreover, SBL  achieves sparsity by employing a Gaussian zero-mean prior on the parameters $X$, where the variance of each parameter is considered a hyperparameter \cite{pml2Book, Tipp01}. This work employs the expectation maximization (EM) scheme to estimate the hyperparameters, where the expectation step calculates the expectation of the log-likelihood with respect to the current parameters (E-step).  The following maximization step updates the parameters to maximize the expectation of the log-likelihood (M-step). As the EM iterations proceed, the algorithm identifies parameters relevant to explaining the observed data. The corresponding prior variances tend to approach zero for parameters not supported by data (i.e., irrelevant to the model). When a variance becomes small, the associated prior distribution effectively becomes a delta function centered at zero, "switching off" those parameters from the model (for a comprehensive explanation of this method, see \cite[Sec~15.2.8]{pml2Book}).

This work aims to determine the sparse frequencies that represent each slice simultaneously without limiting the phase for each frequency. The amplitude of frequency $f_{m,\ell}$, for $m\geq1$ and slice $\ell$ is determined by $\sqrt{x_{m,\ell}^2+x_{m+M,\ell}^2}$, whereas the phase is determined by $\angle(x_{m,\ell},x_{m+M,\ell})$.\footnote{The notation $\angle(a,b)$ denotes the phase of the complex number $a+ib.$} The prior is shared over all slices because the frequencies are assumed to behave similarly over all slices. However, this method should not influence the phase because it is independent of the separate slices. The latter is a modeling choice because the time between capturing the slices is unknown, making it impossible to know the phase of one slice in relation to the next. Thus, the prior should be uniform with respect to the phases. This work introduces a slice-independent variance $\alpha_m>0$ to achieve this uniformity . The prior $p_0(x_{0,\ell},\alpha_0)=\mathcal{N}(0,\alpha_0)$ was selected for frequency $f_0=0$ and   
\begin{equation}
    p_m(x_{m,\ell}\mid\alpha_m) = p_{m+M}(x_{m+M,\ell}\mid\alpha_m)=\mathcal{N}(0,\alpha_m)
\end{equation}
for $m>0$. By setting $\mathbf{\gamma} = [\alpha_0,\alpha_1,\ldots,\alpha_M,\alpha_1,\ldots,\alpha_M]^T$ and $\Gamma = \text{diag}(\gamma)$, the prior of $X$ simplifies to $L$ independent multivariate normal distribution for each slice $L$
\begin{equation}\label{eq:prior}
    p(X\mid\mathbf{\alpha})=\prod_{\ell=1}^L \prod_{m=0}^{2M+1}p_m(x_{m,\ell}\mid\gamma_m) = \prod_{\ell=1}^L\mathcal{N}(0,\Gamma).
\end{equation}
Combining the Gaussian prior and likelihood yields a Gaussian posterior $p(X|Y,\alpha,\sigma^2)$ with the mean and covariance given by
\begin{align}\label{eq:post_mu_sig}
\begin{split}
    \mu_x &= \Gamma A^T \Sigma_y^{-1}Y\\
    \Sigma_x &= \left(\frac{1}{\sigma^2}A^TA+\Gamma^{-1}\right)^{-1},
\end{split}
\end{align}%
where $\Sigma_y=\sigma^2I_N+A\Gamma A^T$ denotes the data covariance. This formulation reveals that the sparsity of $\mu_x$ controls the row sparsity. If entry $\alpha_m = 0$ for $m>0$, it follows that $\gamma_m=\gamma_{m+M}=0$, ensuring the posterior satisfying $p(x_m=x_{m+M}=0|Y,\alpha_m=0)=1$. This approach makes both rows of $X$ corresponding to frequency $f_m$ zero.

\subsubsection{Sparse Bayesian learning and hyperparameter estimation}
The SBL approach relies on estimating the hyperparameters $\alpha_m$ and $\sigma^2$ with empirical Bayes, setting the hyperparameters that maximize the marginal likelihood.  
The marginal likelihood $p(Y|\alpha,\sigma^2)$ is calculated by treating the amplitudes $X$ as a nuisance parameter and marginalizing over them \cite{MSBL2007, SBL_DOA}, obtaining the following:
\begin{equation}
    p(Y|\alpha,\sigma^2)=\int p(y|X,\sigma^2)p(X|\alpha)dX=\frac{\exp\left(-\text{tr}(Y^T\Sigma_y^{-1}Y)\right)}{\left(\pi^N \det\Sigma_y\right)^L}.
\end{equation}
The aim is to maximize the log marginal likelihood, obtaining the following cost function:
\begin{equation}
    \mathcal{L}(\alpha,\sigma^2)=-\text{tr}(Y^T\Sigma_y^{-1}Y)-L \log \det \Sigma_y \propto \log p(Y|\alpha,\sigma^2).
\end{equation}
Next, this cost function is derived with respect to $\alpha_m$. Note that $\alpha_0$ only occurs once in the diagonal matrix $\Gamma$, whereas  $\alpha_m$ for $m>0$ occurs at $\Gamma_{m,m}$ and $\Gamma_{m+M,m+M}$. Here, the focus is on $m>0$, as the derivative with respect to $\alpha_0$ can be calculated analogously. Hence, for $m>0$, 
\begin{equation}
\frac{\partial \Sigma_y^{-1}}{\partial \alpha_m}=-\Sigma_y^{-1}\frac{\partial \Sigma_y}{\partial \alpha_m} \Sigma_y^{-1} = -\Sigma_y^{-1}a_m a_m^T \Sigma_y^{-1}-\Sigma_y^{-1}a_{m+M} a_{m+M}^T \Sigma_y^{-1}
\end{equation}
\begin{equation}
    \frac{\partial \log \det\left(\Sigma_y\right)}{\partial \alpha_m} = \text{tr}\left(\Sigma_y^{-1}\frac{\partial \Sigma_y}{\partial \alpha_m}\right)=a_{m}^T\Sigma_y^{-1}a_{m}+a_{m+M}^T\Sigma_y^{-1}a_{m+M}
\end{equation}
yields the following derivative: 
\begin{align}\label{eq:cost_derivative}
\begin{split}
    \frac{d\mathcal{L}(\alpha,\sigma^2)}{d\alpha_m}&=\text{tr}\left(Y^T\Sigma_y^{-1}a_m a_m^T \Sigma_y^{-1} Y + Y^T\Sigma_y^{-1}a_{m+M} a_{m+M}^T \Sigma_y^{-1} Y \right) \\
    &-L\left(a_{m}^T\Sigma_y^{-1}a_{m}+a_{m+M}^T\Sigma_y^{-1}a_{m+M}\right)\\
    &=||Y^T\Sigma_y^{-1}(a_m+a_{m+M})||_2^2\\
    &-L\left(a_{m}^T\Sigma_y^{-1}a_{m}+a_{m+M}^T\Sigma_y^{-1}a_{m+M}\right),
\end{split}
\end{align}
where $a_m$ denotes the $m$th  column vector of $A$ (cf. \cite{SBL_DOA}). The M-step uses the MacKay update rule introduced in \cite{Mackay}, obtained by setting the derivative \eqref{eq:cost_derivative} to zero and using a fixed-point equation. The update rule 
\begin{equation}\label{eq:alpha_update}
\alpha_m^{\text{new}} = \frac{||(\mu_x)_m+(\mu_x)_{m+M}||_2^2}{L\left(1-\frac{1}{\alpha_m^{\text{old}}}\left((\Sigma_x)_{m,m}+(\Sigma_x)_{m+M,m+M}\right)\right)} 
\end{equation}
is analogous to the results in~\cite{Tipp01, MSBL2007}.

For the estimation of the hyperparameter $\sigma^2$, the update rule introduced in \cite[eq~(27)]{SBL_DOA} is employed, as the definition of the hyperparameter $\alpha_m$ does not influence this update. Thus, the update rule becomes 
\begin{equation}\label{eq:update_sigma}
    \left(\sigma^2\right)^{\text{new}}=\frac{1}{N-K}\text{tr}\left((I_N - A_{\mathcal{M}} A_{\mathcal{M}}^+S_y)\right),
\end{equation}
where $S_y=Y Y^T/L$ represents the data sample covariance matrix. The matrix $A_\mathcal{M}=(a_0, a_{m_1},\ldots,a_{m_K}, a_{m_1+M},\ldots,a_{m_K+M})$ consists of column vectors of $A$ with the indices $m_1,\ldots,m_K$ corresponding to the $K$ largest values in $\alpha$. The Moore–Penrose inverse of $A_{\mathcal{M}}$ is denoted by  $A_{\mathcal{M}}^+$. The parameter $K<<M$ can be selected with model-order selection criteria, as described in \cite{SBL_DOA}. T he choice of $K$ did not significantly influence the results. 

\begin{algorithm}
\caption{Calculate the sparse prior}\label{alg:calc_prior}
\begin{algorithmic}[1]
\Statex \hspace{-\algorithmicindent} \textbf{Initialize:} here $\sigma_0^2=0.2, \alpha=\mathbf{1}_{M+1}, \epsilon_{\min}=1e-4$
\While{$\epsilon > \epsilon_{\min}$}
    \State $\mathbf{\gamma} = [\alpha_0,\alpha_1,\ldots,\alpha_M,\alpha_1,\ldots,\alpha_M]^T,\quad \Gamma = \mathop{\mathrm{diag}}(\gamma)$   
    \State Calculate $\mu_x, \Sigma_x$ \hfill with \eqref{eq:post_mu_sig}\;\;
    \State Update $\alpha_m$ \hfill with \eqref{eq:alpha_update}
    \State Update $\sigma^2$ \hfill with \eqref{eq:update_sigma}
    \State $\epsilon=||\alpha^{\mathrm{new}}-\alpha^{\mathrm{old}}||_1 /||\alpha^{\mathrm{old}}||_1$
\EndWhile
\end{algorithmic}
\end{algorithm}

\subsection{Minimizing labeling work}\label{ssec:min_label}
For each heart, there are slices with reliable volumes over time, whereas others have no information about the ventricle volumes. In these slices, practitioners rely on hand-labeled images. Labeling images by hand is tedious and expensive. The sparse priors obtained using the SBL approach introduced above minimize the labeling effort. Because the priors are sparse, only a few complex amplitudes $x_m+ix_{m+M}$ must be found.

Let $J$ be the index set $J\subset\{1,\ldots,N\}=\Omega$, and $P_J\in \mathbb{R}^{N\times N}$ be the projection matrix, projecting all rows $i\in \Omega \setminus J$ to 0. When limited labeled data $J$ exist, the model \eqref{eq:lin_model} reduces to 
\begin{equation*}
P_JY = P_JAX+N,
\end{equation*}
The posterior of $X$ is a Gaussian with mean $\mu_x$ and covariance $\Sigma_x$, similar to \eqref{eq:post_mu_sig},
\begin{align}
    \begin{split}
    \mu^{(J)}_x &= \Gamma A^TP_J^T \left(\Sigma^{(J)}_y\right)^{-1}Y_J\\
    \Sigma^{(J)}_x &= \left(\frac{1}{\sigma^2}A^TP_JA+\Gamma^{-1}\right)^{-1},
    \end{split}
\end{align}
where $\Sigma^{(J)}_y=\sigma^2I_{|J|}+P_J A\Gamma (P_JA)^T$ is the data covariance--this time with the projected matrices and $|J|$ labeled images. The prediction over the full time is obtained by multiplying the full matrix $A$ with the posterior and adding Gaussian noise $N$. A posterior prediction is obtained, normally distributed with mean $\mu_{post_y}^{(J)}$ and covariance $\Sigma_{post_y}^{(J)}$, given by 
\begin{equation}\label{eq:post_y_J}
    \begin{split}
        \mu_{post_y}^{(J)}&= A\mu_x^{(J)}\\
        \Sigma_{post_y}^{(J)}&= \sigma^2I_N + A\Sigma^{(J)}_xA^T.
    \end{split}
\end{equation}
The variance is a satisfactory measure of the distribution spread for a one-dimensional normal distribution . In this multivariate case, a few possible measures exist \cite{Measure_spread08}. The predominant considerations in the literature are the trace of the covariance matrix $\text{tr}(\Sigma_{post_y}^{(J)})$~\cite{dumbgen_tylers_1998,Visuri_trace03}, the determinant $\det (\Sigma_{post_y}^{(J)})$~\cite{duembgen05det,sal06det}, or the largest eigenvalue $\lambda_1(\Sigma_{post_y}^{(J)})$~ \cite{Randles00EV}. The trace measures the total variance but does not account for the correlation. Nonetheless, the trace  is computationally inexpensive. The determinant of the covariance matrix, as proposed in \cite{Measure_spread08}, offers an alternative but may encounter numerical problems, particularly as $\det(\Sigma)$ approaches 0 in high-dimensional settings with small eigenvalues. A third option is to employ the largest eigenvalue of $\Sigma$, corresponding to the variance along the first principal component. This approach  is also computationally more demanding than applying the trace of the covariance matrix.

This method results in the following optimization problem: Let $S:\mathcal{S}_N\rightarrow\mathbb{R}^+_0$ be a measure for the spread of a distribution, where $\mathcal{S}_N$ denotes the set of real-valued symmetric positive definite $N\times N$ matrices. For $k$ labeled images, we want to determine 
\begin{equation}\label{eq:opt_problem}
J^*:=\arg\min_{|J|=k} S \left(\Sigma_{post_y}^{(J)}\right)
\end{equation}
the indices $J^*$ of the images, which need to be labeled to minimize the spread. 

It might be desirable to let a practitioner iteratively label the next optimal image until practitioners are satisfied with the results. This approach leads  to a slightly different optimization objective, lending itself to applying greedy algorithms. For a given index set $J\subset\{1,\ldots,N\}$, find 
\begin{equation}\label{eq:greedy_optimum}
j^*:=\arg\min_{j\in\{1,\ldots,N\}} S \left(\Sigma_{post_y}^{(J\cup j)}\right).
\end{equation}
T he selected images $J^*$ and $j^*$ do not depend on the measurements at those slices. Thus, the indices can be computed before the labeling process. However, they strongly depend on the computed sparse frequencies. The procedure is presented in  Algorithm~\ref{alg:label_efficiency}.

\subsection{Optimality of the greedy algorithm}\label{ssec:optimality_greedy}
Finding the optimal solution of \eqref{eq:opt_problem} is difficult  due to the combinatorial explosion as $k$ increases. There are already more than 19 billion combinations to draw five images from $N=300$ time steps. Thus, iteratively solving \eqref{eq:greedy_optimum} using a greedy approach is a feasible alternative. The objective function is first reformulated to establish performance guarantees. Let $J\subseteq\Omega =\{1,\ldots , N\}$ and $\Sigma_y:= \sigma^2 I_N+A\Gamma A^T$ the data covariance, we define the objective function as follows: 
\begin{equation}
\begin{split}\label{eq:obj_mono_nondec}
f(J)=&\text{tr}\left(\Sigma_y\right)-\text{tr}\left(\Sigma_{post_y}^{(J)}\right)\\
=&\text{tr}(\Sigma_y)-\text{tr}\left(\sigma^2 I_N+A\left(\frac{1}{\sigma^2}A^TP_JA+\Gamma^{-1}\right)^{-1}A^T\right).
\end{split}
\end{equation} 
Maximizing this function is equivalent to minimizing $\text{tr}\left(\Sigma_{post_y}^{(J)}\right)$. For brevity, this work is limited to investigating the trace because this measure performed best in the experiments.  The following definitions are introduced:
\begin{defi}
Let $\Omega$ be a finite set and $2^{\Omega}$ be the power set of $\Omega$. A set function $f:2^{\Omega}\rightarrow \mathbb{R}$ is \underline {monotone non  decreasing} if $f(\setX)\leq f(\setY )$ for all $\setX\subseteq \setY \subseteq \Omega$.
\end{defi}
\begin{defi}
    The \underline{marginal gain} of adding element $j$ to set $\setX$ is defined as follows:
    \begin{equation}
    f_j(\setX):= f(\setX\cup {j})-f(\setX).
    \end{equation}
\end{defi}
\begin{defi}
A set function $f:2^{\Omega}\rightarrow \mathbb{R}$ is \underline {submodular} if
\begin{equation}
f_j(\setX)\geq f_j(\setY )
\end{equation}
for every subset $\setX\subseteq \setY \subset \Omega$ and every $j\in \Omega\setminus \setY $.
\end{defi}

The authors of \cite{nemhauser_analysis_1978} demonstrated  that the greedy algorithm generates a solution at least $(1-1/e)$ of the optimal solution for submodular set functions. The objective function \eqref{eq:obj_mono_nondec} is not submodular, as Figure~\ref{fig:wsc_hist} presents counterexamples. However, a weaker bound exists for \eqref{eq:obj_mono_nondec}. This work follows the definition from \cite{hashemi19a}.
\begin{defi}\label{defi:wsc}
The \underline {weak-submodularity constant} $c_f$ of a monotone non decreasing  function $f$ is defined as follows: 
    \begin{equation}
        c_f := \max_{(\setX ,\setY ,i)\in \tilde{\Omega}} f_i(\setY)/f_i(\setX),
    \end{equation}
    where $\tilde{\Omega}=\{(\setX , \setY , i)|\setX \subseteq \setY \subset \Omega, i\in \Omega\setminus\setY\}.$
\end{defi}
Note, that $f$ is submodular if and only if $c_f\leq1$ and is weakly submodular if and only if $c_f$ is bounded.\\
For $J\subset \Omega$ and $i\in \Omega$ the marginal gain of the objective function \eqref{eq:obj_mono_nondec} simplifies to 
\begin{equation}
    \begin{split}
        f_i(J)=& f(J\cup \{i\}) - f(J)\\
        =& \text{tr}\left(\Sigma_{post_y}^{(J)}\right)-\text{tr}\left(\Sigma_{post_y}^{(J\cup \{i\})}\right).
    \end{split}
\end{equation}
Moreover, we can prove the following lemma.
\begin{lem}\label{lem:proof}
For $f_i$ defined in \eqref{eq:obj_mono_nondec}, it holds $f_i(J)>0$ for all $J\subset \Omega$ and $i\in \Omega\setminus J$.
\end{lem}
\begin{proof}
If $J\subset \Omega$ and $i\in \Omega\setminus J$, it holds that
\begin{equation}\label{eq:proof_mono}
\begin{split}
        f_i(J)=&\text{tr}\left(\Sigma_{post_y}^{(J)}\right)-\text{tr}\left(\Sigma_{post_y}^{(J\cup \{i\})}\right)\\
        =& \text{tr}\left(A\left(\frac{1}{\sigma^2}A^TP_{J}A+\Gamma^{-1}\right)^{-1}A^T\right) - \text{tr}\left(A\left(\frac{1}{\sigma^2}A^TP_{J\cup \{i\}}A+\Gamma^{-1}\right)^{-1}A^T\right)\\
        =& \text{tr}\left(A\left(\left(\frac{1}{\sigma^2}A^TP_{J}A+\Gamma^{-1}\right)^{-1}-\left(\frac{1}{\sigma^2}A^TP_{J\cup \{i\}}A+\Gamma^{-1}\right)^{-1}\right)A^T\right).\\
\end{split}
\end{equation}
This work defines $D:=1/\sigma^2 A^TP_{J}A+\Gamma^{-1}$, which is symmetric positive definite because $\Gamma$ is a diagonal matrix with positive diagonal entries. For $D$ and $e_i$, the standard basis vector in the $i$th direction, note that 
\begin{equation}\label{eq:proof_mono1}
    \begin{split}
    \frac{1}{\sigma^2}A^TP_{J\cup \{i\}}A+\Gamma^{-1} =& D+\frac{1}{\sigma^2}A^T P_{\{i\} A}\\
    =&D+A^Te_i e_i^TA
    \end{split}
\end{equation}
With this formulation \eqref{eq:proof_mono1} and the Sherman--Morrison formula \cite{Barlett_sherman_morrison}, \eqref{eq:proof_mono} can be written as follows:
\begin{equation}\label{eq:lemma_proof}
    \begin{split}
        f_i(J)=&\text{tr}\left(A \left(D^{-1}-\left(D+A^Te_i e_i^TA\right)^{-1}\right) A^T\right)\\
        \underset{S.-M.}{=}&\text{tr}\left(A\left((1+\frac{1}{\sigma^2}e_i^TAD^{-1}A^Te_i)^{-1}\frac{1}{\sigma^2}D^{-1}A^Te_i e_i^T AD^{-1}\right)A^T\right)\\
        =&\left(1+\frac{1}{\sigma^2}e_i^TAD^{-1}A^Te_i\right)^{-1}\left(\frac{1}{\sigma^2} e_i^T AD^{-1}A^TAD^{-1}A^Te_i\right),
    \end{split}
\end{equation}
where $\langle\cdot,\cdot\rangle_D$ can be written for the weighted inner product because $D$ is a symmetric positive definite matrix. The marginal gain in \eqref{eq:lemma_proof} simplifies to  
\begin{equation}
    \begin{split}
        f_i(J)=&\left(1+\frac{1}{\sigma^2}e_i^TAD^{-1}A^Te_i\right)^{-1}\left(\frac{1}{\sigma^2} e_i^T AD^{-1}A^TAD^{-1}A^Te_i\right)\\
        =&\frac{1}{\sigma^2}\langle AD^{-1}A^Te_i,AD^{-1}A^Te_i\rangle /\left(1+\frac{1}{\sigma^2}\langle A^Te_i,A^Te_i\rangle_{D^{-1}}\right)\\
        >&0.
    \end{split}
\end{equation}
This result holds because the numerator is zero if and only if a zero row $i$ exists in $A$, but by construction, the first entry of each row is 1. %
\end{proof}
\begin{thm}
The objective function $f(J)$, as defined in \eqref{eq:obj_mono_nondec}, is a weakly submodular set function.
\end{thm}
\begin{proof}
From Lemma~\ref{lem:proof}, it follows directly that $f(J)$ is a monotone non-decreasing function. It is left to be shown, that the weak-submodularity constant $c_f$ is bounded.\\
As $\Omega$ is a finite set, there exists 
\begin{equation}
    c_{\text max}:=\max_{\setX\subset\Omega, i\in \Omega\setminus\setX}f_i(\setX)
\end{equation}
and 
\begin{equation}
    c_{\text min}:=\min_{\setX\subset\Omega, i\in \Omega\setminus\setX}f_i(\setX).
\end{equation}
We know from Lemma~\ref{lem:proof}, that $0<c_{\text min}\leq c_{\text max}$; thus, $c_b=c_{\text max}/c_{\text min}$ exists, and the weak-submodularity constant $c_f$ is bounded with 
\begin{equation}
    c_f = \max_{(\setX ,\setY ,i)\in \tilde{\Omega}} f_i(\setY)/f_i(\setX)\leq c_b,
\end{equation}
for $\tilde{\Omega}=\{(\setX , \setY , i)|\setX \subseteq \setY \subset \Omega, i\in \Omega\setminus\setY\}$.
\end{proof}
Based on \cite{kempe_weak_submodular11, hashemi19a}, the solution obtained by the greedy algorithm is guaranteed to be within a factor of $(1-e^{-1/c_f})$ of the optimal solution $f(J^*).$ Section~\ref{sec:results} empirically examines the weak-submodularity constant. 

\section{Application}\label{sec:application}
This application involves real-time MRI scans of univentricular hearts, with videos of $N=300$ images in each slice through the heart. Although the automatic segmentation in the intermediate slices is very reliable, it contains no information on the outer slices. This work is interested in the ventricle volume for each slice over time. Combining the results of Sections~\ref{ssec:find_freqs} and \ref{ssec:min_label} can minimize the labeling effort in the outer slices. This work builds a linear model \eqref{eq:lin_model} to describe the volume over the slices and time. Section~\ref{ssec:find_freqs} employs the SBL algorithm to determine the prior empirically. The SBL algorithm ensures that the prior is sparse. Thus, the prior for many frequencies is close to a delta distribution around 0. The distributions are normal with variances of $\alpha_m\approx 0$. Moreover, $\alpha_m$ can be set to 0 for small $\alpha_m$, and the corresponding columns can be removed from $A$ and rows removed from $X$ to circumvent numerical problems and reduce the computational expense of the problem. A practical method to set a threshold is $\epsilon_{thresh} = 0.01 \cdot \max_{m\in \{1,\ldots, M\}}(\alpha_m)$, as frequencies with amplitudes orders of magnitudes smaller than the most significant frequency have a negligible effect on the final volume.

A prior that assigns weight only to a few frequencies is beneficial when predicting slices without information about the volume because only the phase and amplitude of the most influential frequencies must be fit. Thus, only a few labeled images suffice to predict a whole slice. The sequence of images can be optimized for labeling by solving \eqref{eq:opt_problem} for a fixed number of images. This problem is difficult  due to the combinatorial explosion for larger numbers. Practically, \eqref{eq:greedy_optimum} can be iteratively solved, and the best slice can be consecutively labeled until the practitioner is satisfied with the results. Section~\ref{ssec:optimality_greedy} proves that the solution of this approach is at most a factor of $(1-e^{-1/c_f})$ away from the optimal solution to \eqref{eq:opt_problem}.

\begin{algorithm}
\caption{Efficient labeling and prediction}\label{alg:label_efficiency}
\begin{algorithmic}[1]
\Statex \hspace{-\algorithmicindent} \textbf{Require:} $\Gamma = \mathop{\mathrm{diag}}([\alpha_0,\alpha_1,\ldots,\alpha_M,\alpha_1,\ldots,\alpha_M])$ and $\sigma^2$ from Algorithm \ref{alg:calc_prior} 
\Statex \hspace{-\algorithmicindent} \textbf{Initialize:} $J=\{\,\}$
\While{insufficient prediction}
\State $j^*:=\arg\min_{j\in\{1,\ldots,N\}} S \left(\Sigma_{post_y}^{(J\cup j)}\right)$ \hfill with $S$ from \eqref{eq:spread_measure}
\State Manually label $j^*$
\State $J=J\cup j^*$
\State Calculate posterior        $\mu_{post_y}^{(J)}$, $\Sigma_{post_y}^{(J)}$  \hfill with \eqref{eq:post_y_J} 
\EndWhile
\end{algorithmic}
\end{algorithm}

\section{Results}\label{sec:results}
The MRI scans are from patients with univentricular hearts. This condition  makes automatic segmentation incredibly challenging because their hearts usually differ significantly from the geometry of the hearts of other patients. Therefore, it is not possible to use publicly available datasets. The data contain  MRI scans of two patients. The inner five slices for Heart~1 and the inner four slices for Heart~2 can be segmented reliably using a U-Net approach similar to \cite{isensee_nnu-net_2021}. A medical expert confirmed the segmentation for those slices. The outer slices could not be segmented reliably, with mostly zero predictions. 
\begin{figure}[h]
    \centering
        \centering
        \input{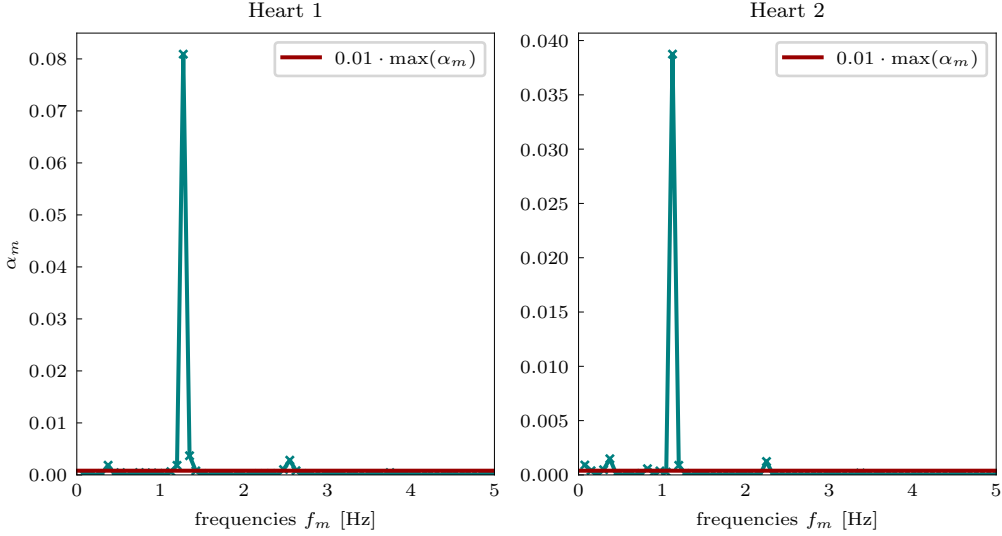}
    \caption{Values for $\alpha_m$, where $m>0$, after the expectation maximization procedure of the SBL algorithm. The largest peak is at the heart rate, with smaller peaks at double the heart rate. The respiratory rate is about $0.25$~Hz.}
    \label{fig:alpha_m}
\end{figure}

\subsection{Finding sparse frequencies}
The frequencies $f_m=3m/40$~Hz for $m=0,\ldots100$ were employed to cover a broad range of possible frequencies. Although one of the ventricle volumes (Heart~1) is visibly more irregular than the other, the SBL algorithm reliably found sparse frequencies for both hearts. Figure~\ref{fig:alpha_m} illustrates  the values for $\alpha_m$, $m>0$, after the EM procedure of the SBL algorithm. For both hearts, the algorithm found seven frequencies higher than the threshold $0.01\cdot\max_{m>0}(\alpha_m)$. By far, the most significant influence is the heart frequency. Minor influences include the respiratory rate at about $0.3$~Hz and double the heart rate (i.e., , higher harmonics). Figure~\ref{fig:posterior} depicts the posterior of the slices. The broader posterior for Heart~1 captures the more irregular heartbeat.
\begin{figure}[h!]
    \centering
    \begin{subfigure}{0.95\textwidth}
        \centering
        \includegraphics[width=\textwidth]{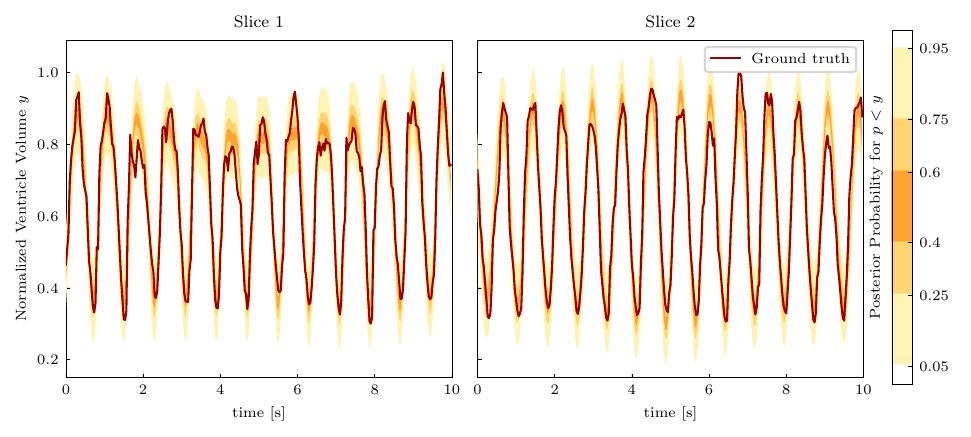}
        \caption{Heart 1}
    \end{subfigure}
    \\
    \begin{subfigure}{0.95\textwidth}
        \centering
        \includegraphics[width=\textwidth]{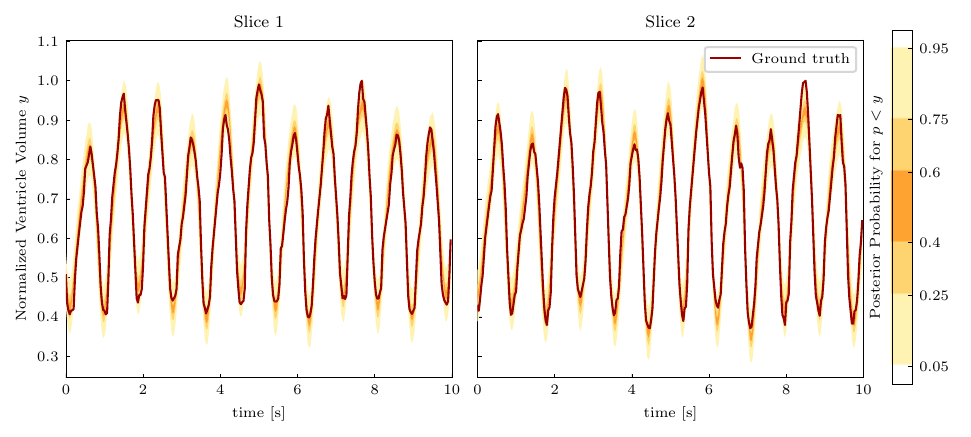}
        \caption{Heart 2}
    \end{subfigure}
    \caption{Posterior distribution and ground-truth data for two slices of two hearts. (a) Heart 1 is much less regular. The posterior distribution reflects this, with a broader posterior.}
    \label{fig:posterior}
\end{figure}%

\subsection{Minimizing labeling work}
The labeling effort on the other slices can be minimized after empirically determining the prior on the good intermediate slices. We used a greedy algorithm to search for the best next point to label. As objective functions, this work used the following: 
\begin{align}\label{eq:spread_measure}
    \begin{split}
        S\left(\Sigma_{post_y}^{(J)}\right) &= \text{tr}\left(\Sigma_{post_y}^{(J)}\right)\\
        S\left(\Sigma_{post_y}^{(J)}\right) &= \sqrt[N]{\det \left(\Sigma_{post_y}^{(J)}\right)}\\
        S\left(\Sigma_{post_y}^{(J)}\right) &= \lambda_1\left(\Sigma_{post_y}^{(J)}\right).
    \end{split}
\end{align}
The eigenvalues of the matrix became too small for a numerical calculation of $\det\left(\Sigma_{post_y}^{(J)}\right)$, which yields a value of 0. Therefore, the property that the determinant is the product of the eigenvalues was applied to compute the following:
\begin{equation}
            S\left(\Sigma_{post_y}^{(J)}\right) = \sqrt[N]{\det \left(\Sigma_{post_y}^{(J)}\right)} = \prod_{i=1}^N \sqrt[N]{\lambda_i \left(\Sigma_{post_y}^{(J)}\right)}.
\end{equation}
Section~\ref{ssec:optimality_greedy} demonstrates that, for the trace $\text{tr} \left(\Sigma_{post_y}^{(J)}\right)$, the solution obtained by the greedy algorithm is guaranteed to be within a factor of $(1-e^{-1/c_f})$ of the optimal solution, where $c_f$ denotes the weak-submodularity constant 
\begin{equation}
    c_f := \max_{(\setX ,\setY ,i)\in \tilde{\Omega}} f_i(\setY)/f_i(\setX),
\end{equation}
with $\tilde{\Omega}=\{(\setX , \setY , i)|\setX \subseteq \setY \subset \Omega, i\in \Omega\setminus\setY\}$ for $\Omega=\{1,\ldots,300\}$. Next, 30,000 samples were uniformly drawn for $i\in \Omega$, $\setY \in \Omega\setminus\{i\}$ and $\setX\subset\setY$ to estimate $c_f$, and the following is calculated:
\begin{equation}\label{eq:estimated_wsc}
    c_{(\setX ,\setY ,i)}:=f_i(\setY)/f_i(\setX)
\end{equation}
The highest value was $\approx2.52$, whereas only $1.2\%$ of the samples were greater than 1 (see Figure~\ref{fig:wsc_hist}). The values greater than 1 prove that the function is not a submodular set function, but the moderate value could suggest that $c_f$ is within this order of magnitude.

\begin{figure}[h]
    \centering
        \centering
        \includegraphics[width=0.7\textwidth]{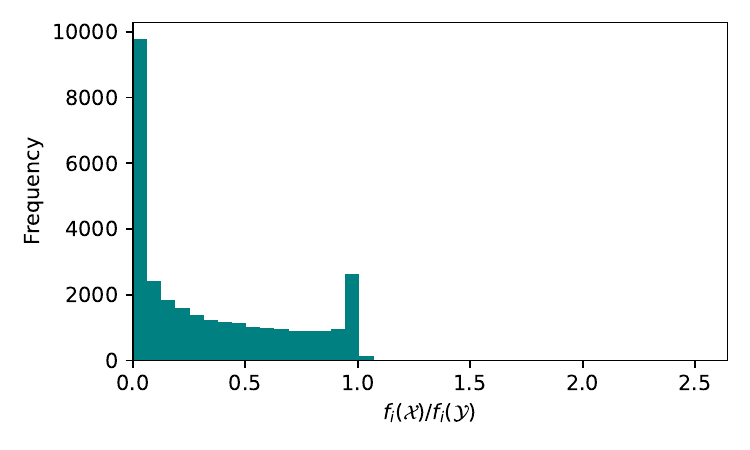}
    \caption{Histogram of 30,000 random draws of $c_{(\setX ,\setY ,i)}$, as defined in \eqref{eq:estimated_wsc}, to estimate the weak-submodularity constant. The highest obtained value is $\approx2.52$,  and only $1.2\%$ of the samples are greater than 1.}
    \label{fig:wsc_hist}
\end{figure}%
We used jackknife resampling to evaluate the performance of the approach on the patient data because additional thoroughly segmented slices from the respective hearts were unavailable. Thus, slices were left out to determine the prior with sampling from the left-out slice.  Despite the presence of visual obstructions, such as valves that may appear differently in certain heart diseases and pose challenges for neural network approaches to segment the images accurately, the most significant frequencies affecting the volume over time, such as respiratory and heart rates, are expected to remain consistent in those slices. Hence, the posterior predictions can be compared to the segmented images as ground truth.

\begin{figure*}[ht]
        \centering

        \begin{subfigure}[b]{0.49\textwidth}
            \centering
            \includegraphics[width=\textwidth]{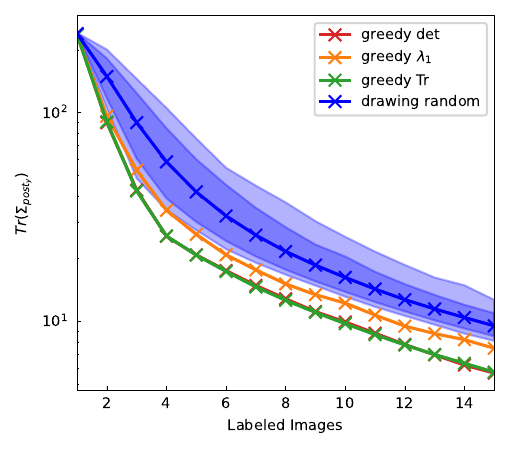}
            \caption{Trace Heart 1}%
            \label{fig:trace_heart_1}
        \end{subfigure}
        \hfill
        \begin{subfigure}[b]{0.49\textwidth}
            \centering
            \includegraphics[width=\textwidth]{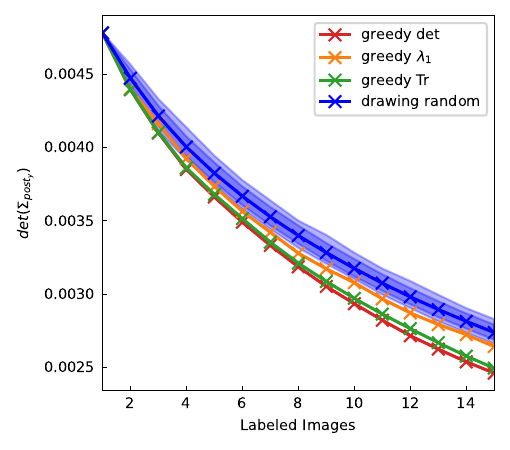}
            \caption{Determinant Heart 1}%
            \label{fig:det_heart_2}
        \end{subfigure}
        \begin{subfigure}[b]{0.49\textwidth}
            \centering
            \includegraphics[width=\textwidth]{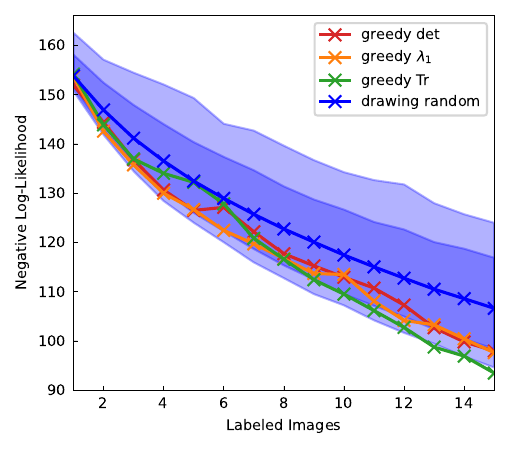}
            \caption{nll Heart 1}%
            \label{fig:nll_heart_1}
        \end{subfigure}
        \hfill
        \begin{subfigure}[b]{0.49\textwidth}  
            \centering 
            \includegraphics[width=\textwidth]{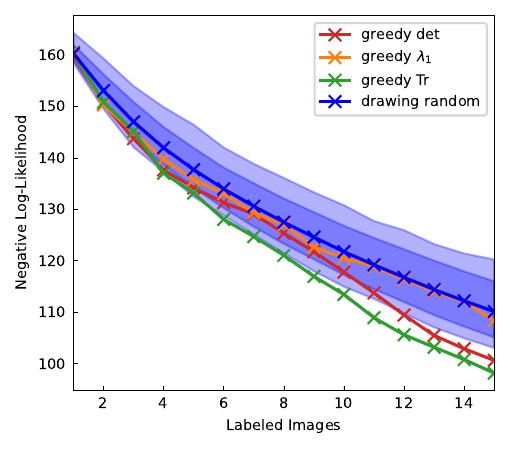}
            \caption{nll Heart 2}%
            \label{fig:nll_heart_2}
        \end{subfigure}

    \caption{Results of optimizing the posterior covariance. (a, b) Resulting covariance of the greedy approaches with 10,000 random draws. Shaded areas display the range of the random draws, where darker blue marks the inner 90\%,  and the blue curve represents the mean for (a) covariance trace, (b) determinant, and (c, d) resulting negative log-likelihood on Hearts~1 and 2. The wider band on Heart 1 reflects the higher irregularity in the heartbeat in this heart.}\label{fig:cov_and_nll}
\end{figure*}%
All three objective functions aim to reduce the posterior distribution spread. Figure~\ref{fig:trace_heart_1} and~\ref{fig:det_heart_2} demonstrate that each objective function produces narrower posterior distributions than  random sampling, suggesting that the greedy algorithm is a sensible approach.  Appendix~\ref{sec:appendix} compares all measures on both hearts.  The posterior distributions found by minimizing the trace and determinant of the posterior covariance matrix perform similarly, minimizing the largest eigenvalue yields a wider posterior. The first eigenvalue seems insufficient to describe the full spread of the posterior. More critical than the spread of the posterior distribution is for the posterior to represent the unknown data accurately. The jackknife sampling method enables a comparison of the posterior distributions with the ground-truth data, as illustrated  in Figure~\ref{fig:nll_heart_1} and \ref{fig:nll_heart_2} . The negative log-likelihood of the left-out slices using the greedy approach is compared with that of 10,000 random samples. 

Among the three objective functions, $S\left(\Sigma_{post_y}^{(J)}\right) = \text{tr}\left(\Sigma_{post_y}^{(J)}\right)$ achieves the best results, as it is the only method to surpass random sampling across all samples with 15 labeled images on both hearts. The determinant yields slightly worse results but still outperforms all random samples for more than 12 labeled images in Heart~2 and over 95\% of the random samples in Heart~1 for 14 and 15 labeled images. Using the largest eigenvalue as an objective function yields the worst results, barely outperforming the average random sampling approach on Heart~2. This finding  is also consistent with the observation that minimizing the largest eigenvalue fails to minimize the spread of the posterior to the same degree as the other two objectives.

Overall, the most important benefit of using these greedy algorithms is preventing the random selection of inefficient sequences, especially in the more irregular heart (Heart~1).  The broader negative log-likelihood range of the random samples indicates the importance of an informed selection of labeled images. Drawing seven images according to the trace of the posterior distribution outperforms the worst random sample with 15 labeled images. For Heart~2, only nine labeled images are needed to outperform the worst random sample with 15 images. The importance of avoiding poor sequences can also be visually observed. The posterior distribution captures the ground truth for only five sample points selected with the greedy trace approach (see Figure~\ref{fig:samples_1} and~\ref{fig:samples_2}). In contrast , when selecting the sample images with a greedy worst approach, the posterior does not fit the data. However, the model also represents this uncertainty in a wide posterior distribution, providing a clear warning signal for the practitioner not to trust the predictions (Figure~\ref{fig:greedy_worst}).

\begin{figure*}[ht]
        \centering
        \begin{subfigure}[b]{0.49\textwidth}   
            \centering 
            \includegraphics[width=\textwidth]{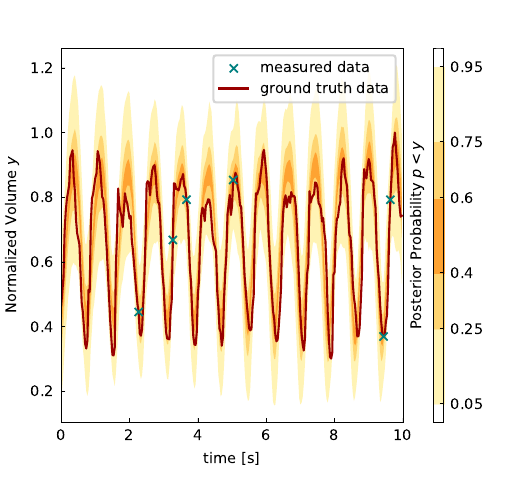}
            \caption{Five labeled images (Heart 1)}%
            \label{fig:samples_1}
        \end{subfigure}
        \hfill
        \begin{subfigure}[b]{0.49\textwidth}   
            \centering 
            \includegraphics[width=\textwidth]{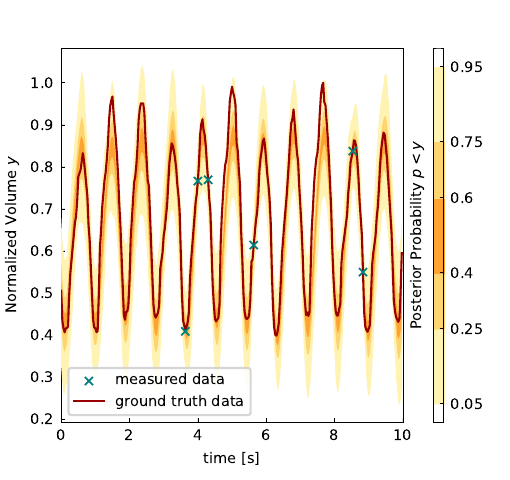}
            \caption{Five  labeled images (Heart 2)} %
            \label{fig:samples_2}
        \end{subfigure}
        \begin{subfigure}[b]{0.49\textwidth}  
            \centering 
            \includegraphics[width=\textwidth]{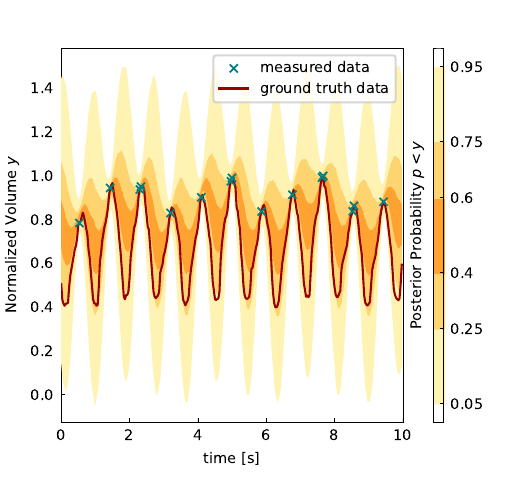}
            \caption{Greedy worst}%
            \label{fig:greedy_worst}
        \end{subfigure}
        \caption{Predictive posterior for measured data depicted in the shaded yellow area for labeled images from unknown ground-truth data . The measured data is marked with a green x. (c) Results  of 15 labeled images selected by a greedy worst approach on Heart 2. Compared to the result of five selected images by a greedy algorithm minimizing the trace on Heart 1 (a) or Heart 2 (b), the posterior variance is significantly smaller in (a) and (b) with fewer labeled images. The posterior also corresponds more accurately to the ground-truth data in those examples.} 
        \label{fig:good_vs_bad_samples}
    \end{figure*}%

\section{Conclusion}\label{sec:conclusion}
The SBL algorithm identifies the dominant frequencies in ventricle volumes over time from real-time MRI scans. Sparse frequency priors can efficiently label slices where automatic segmentation is inadequate. A greedy approach can be applied to select the optimal image for manual segmentation to minimize an objective function related to the posterior covariance spread. In this application, using the trace of the posterior covariance as an objective function is the least computationally expensive  and yields the best results. Performance can be guaranteed for the greedy algorithm on the trace of the posterior covariance. The greedy approach stays within a factor of $(1-e^{-1/c_f})$ of the optimal solution. Empirically, the highest value  was $\approx2.5$, suggesting that $c_f$ might also be in this order of magnitude. Testing this approach on real-life patient data demonstrated  that the selection procedure minimized the posterior covariance spread. Moreover, the resulting models also accurately predicted the unknown data. Specifically, the model using 15 labeled images selected based on the trace of the posterior covariance consistently outperformed all models using randomly selected images, achieving the lowest negative log-likelihood for Hearts~1 and 2  across 10,000 random samples.

\section*{Acknowledgments}
This work was performed  as  part  of  the  Helmholtz School for Data Science in Life, Earth and Energy (HDS-LEE) and received funding from the Helmholtz Association of German Research Centres. This publication is partly supported by the Alexander von Humboldt Foundation and the King Abdullah University of Science and Technology (KAUST) Office of Sponsored Research (OSR) under Award No. OSR-2019-CRG8-4033. \\
We would like to thank the whole Hypofon study team: Nicole Müller, Julian Alexander Härtel, Jan Schmitz, Ute Baur, Melanie von der Wiesche, Iris Rieger, Jon von Stritzky, Christopher Hart, Janina Bros, Benedikt Seeger, Emily Zollmann, Marijke Grau, Boris Dragutinovic, Laura-Maria de Boni, Jan-Niklas Hönemann, Wilhelm Bloch, Daniel Aeschbach, Eva-Maria Elmenhorst, Ulrike Herberg, Alena Hess, Moritz Schumann, Tobias Kratz, Jens Jordan and Johannes Breuer and the particpants for the unique study data and Stiftung KinderHerz for making this study possible. 

\FloatBarrier
\begin{appendices}

\section{Supplementary Figures}\label{sec:appendix}
This work measures the posterior covariance spread with the trace, determinant, and largest eigenvalue of the posterior covariance matrix (see \eqref{eq:spread_measure}). The greedy algorithm determined which images to label on the three mentioned measures. Figure~\ref{fig:appendix_all_measures} compares the performance of the selection schemes on both hearts for all three measures. 
\begin{figure*}[ht]
        \centering

        \begin{subfigure}[b]{0.45\textwidth}
            \centering
            \includegraphics[width=\textwidth]{figures/trace_heart_1.pdf}
            \caption{Trace Heart 1}%
        \end{subfigure}
        \hfill
        \begin{subfigure}[b]{0.45\textwidth}
            \centering
            \includegraphics[width=\textwidth]{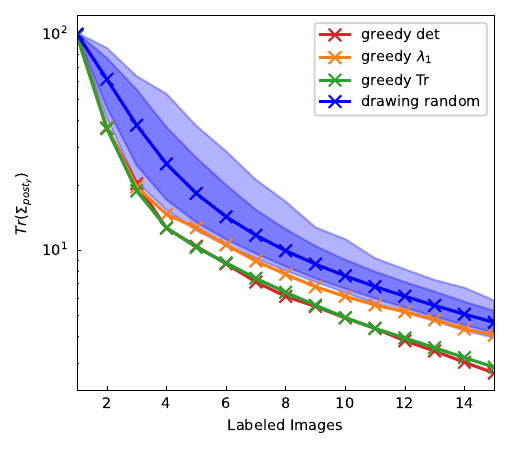}
            \caption{Trace Heart 2}%
        \end{subfigure}
        \begin{subfigure}[b]{0.45\textwidth}
            \centering
            \includegraphics[width=\textwidth]{figures/det_heart_1.pdf}
            \caption{Determinant Heart 1}%
        \end{subfigure}
        \hfill
        \begin{subfigure}[b]{0.45\textwidth}
            \centering
            \includegraphics[width=\textwidth]{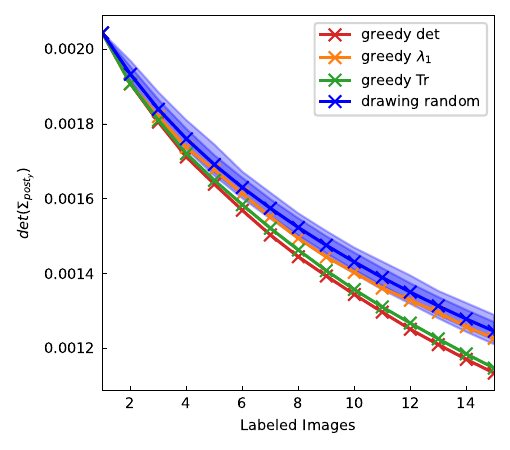}
            \caption{Determinant Heart 2}%
        \end{subfigure}
        \begin{subfigure}[b]{0.45\textwidth}
            \centering
            \includegraphics[width=\textwidth]{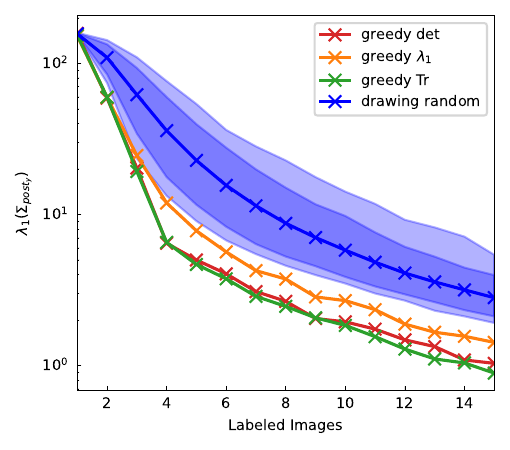}
            \caption{Largest eigenvalue Heart 2}%
        \end{subfigure}
        \hfill
        \begin{subfigure}[b]{0.45\textwidth}
            \centering
            \includegraphics[width=\textwidth]{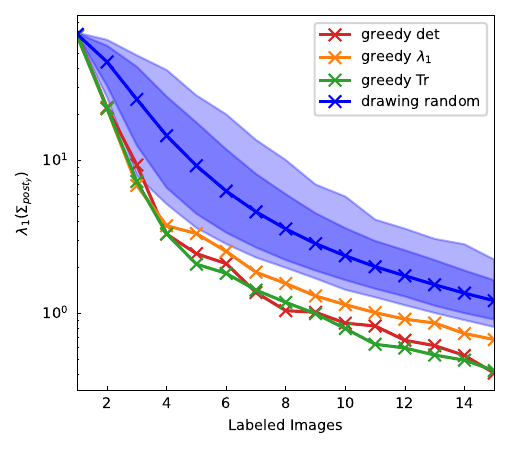}
            \caption{Largest Eigenvalue Heart 2}%
        \end{subfigure}
        \caption{Results of optimizing the posterior covariance, comparing the covariance of the greedy approaches with 10,000 random draws. Shaded areas depict  the range of random draws, where darker blue marks the inner 90\%,  and the blue curve is the mean. For comparison, the largest eigenvalue, determinant, and trace of the posterior covariance are plotted for both hearts.}\label{fig:appendix_all_measures}
\end{figure*}%

\end{appendices}
\FloatBarrier

\bibliography{literature}

\end{document}